\documentclass[journal]{IEEEtran}
\usepackage{xargs}
\usepackage{soul}
\usepackage{cancel}
\usepackage{mathtools}
\usepackage{verbatim}
\usepackage{enumitem}
\usepackage{soul}
\usepackage{graphics}
\usepackage{epsfig} 
\usepackage{amsmath,lipsum} 
\usepackage{amssymb}  
\usepackage{amsthm}
\usepackage{xfrac}
\usepackage[pdfborder={0 0 0}]{hyperref}
\graphicspath{{figures/}}
\newcommand {\matr}[2]{\left[\begin{array}{#1}#2\end{array}\right]}

\newcommand{\x}{{\mathbf{x}}}
\newcommand{\y}{{\mathbf{y}}}
\newcommand{\z}{{\mathbf{z}}}
\renewcommand{\u}{{\mathbf{u}}}
\newcommand{\vv}{{{v}}}
\newcommand{\w}{{\mathbf{w}}}
\renewcommand{\r}{{\mathbf{r}}}
\newcommand{\rr}{{\mathrm{r}}}

\newcommand{\bxi}{{\boldsymbol{\xi}}}
\newcommand{\bet}{{\boldsymbol{\nu}}}
\newcommand{\bx}{{\x}}
\newcommand{\bu}{{\u}}
\newcommand{\by}{{\tilde\y}}

\newcommand{\blambda}{{\boldsymbol{\lambda}}}
\newcommand{\bmu}{{\boldsymbol{\mu}}}

\usepackage[usenames,dvipsnames]{xcolor}
\definecolor{wheat}{rgb}{0.96,0.87,0.70}
\definecolor{mario}{rgb}{0.8,0.8,1}
\definecolor{ivo}{rgb}{1,0.8,0.8}

\usepackage{enumitem}

\newcommandx{\xb}[2][1=n,2=k]{\x_{#1|#2}}

\newcommandx{\zb}[2][1=n,2=k]{\bar\z_{#1|#2}}
\newcommandx{\ub}[2][1=n,2=k]{\u_{#1|#2}}
\newcommandx{\yb}[2][1=n,2=k]{\bar\y_{#1|#2}}
\newcommandx{\vb}[2][1=n,2=k]{\vv_{#1|#2}}
\newcommandx{\wb}[2][1=n,2=k]{\w_{#1|#2}}
\newcommandx{\rx}{\r^{\x}}
\newcommandx{\ru}{\r^{\u}}
\newcommandx{\tb}[2][1=n,2=k]{\tau_{#1|#2}}

\newcommandx{\xt}[2][1=n,2=k]{\tilde\x_{#1|#2}}
\newcommandx{\ut}[2][1=n,2=k]{\tilde\u_{#1|#2}}
\newcommandx{\vt}[2][1=n,2=k]{\tilde\vv_{#1|#2}}
\newcommandx{\tildet}[2][1=n,2=k]{\tilde\tau_{#1|#2}}

\newcommandx{\rbx}[2][1=n,2=k]{\r_{#1|#2}^{\x}}
\newcommandx{\rbu}[2][1=n,2=k]{\r_{#1|#2}^{\u}}

\newcommandx{\hb}[3][1=n,2=k,3={}]{h_{#1}^{#3}}
\newcommandx{\gb}[3][1=n,2=k,3={}]{g_{#1|#2}^{#3}}
\newcommandx{\gbar}[3][1=n,2=k,3={}]{\bar{g}_{#1|#2}^{#3}}
\newcommandx{\xT}[2][1=n,2=k]{\mathcal{X}_{#1|#2}}
\newcommandx{\sigb}[3][1=n,2=k,3={}]{\bar\sigma_{#1|#2}^{#3}}
\usepackage{tikz}
\usepackage{soul}

\newtheorem{Theorem}{Theorem}
\newtheorem{Lemma}{Lemma}
\newtheorem{Proposition}{Proposition}
\newtheorem{Corollary}[Theorem]{Corollary}
\newtheorem{Assumption}{Assumption}

\newtheorem{Remark}{Remark}

\newcommand{\paolor}[1]{\textcolor{black}{#1}}

\begin{document}
	\title{Model Predictive Control with Infeasible Reference Trajectories}
	\author{Ivo~Batkovic$^{1,2}$,~Mohammad~Ali$^2$,~Paolo~Falcone$^{1,3}$,~and~Mario~Zanon$^4$
		
		\thanks{This work was partially supported by the Wallenberg Artificial Intelligence, Autonomous Systems and Software Program (WASP) funded by Knut and Alice Wallenberg Foundation.
		}
		\thanks{$^{1}$ Ivo Batkovic and Paolo Falcone are with the Mechatronics group at the Department of Electrical Engineering, Chalmers University of Technology, Gothenburg, Sweden {\tt\footnotesize \{ivo.batkovic,falcone\}@chalmers.se }}
		\thanks{$^{2}$ Ivo Batkovic, and Mohammad Ali are with the research department at Zenseact AB {\tt\footnotesize \{ivo.batkovic,mohammad.ali\}@zenseact.com}}
		\thanks{$^{3}$ Paolo Falcone is with the Dipartimento di Ingegneria ``Enzo Ferrari'' Universit\`a di Modena e Reggio Emilia, Italy {\tt\footnotesize falcone@unimore.it}}
		\thanks{$^{4}$ Mario Zanon is with the IMT School for Advanced Studies Lucca {\tt\footnotesize mario.zanon@imtlucca.it}}}

	\maketitle
	
	\begin{abstract}
		Model Predictive Control~(MPC) formulations are typically built on the requirement that a feasible reference trajectory is available. In practical settings, however, references that are infeasible with respect to the system dynamics are used for convenience. In this paper, we prove under which conditions an MPC formulation is Input-to-State Stable~(ISS) in closed-loop when an infeasible reference is used, and that with proper terminal conditions, asymptotic stability towards an optimal reference may be achieved. We illustrate the theoretical results with a four-dimensional robotic joint example.
	\end{abstract}
	
	
	\IEEEpeerreviewmaketitle
	
	\section{Introduction}
	Model Predictive Control~(MPC) is widely known as an advanced control technique for nonlinear systems that can handle time-varying references with preview information as well as constraints. In the vast body of literature, standard MPC formulations penalize deviations from a set point (typically the origin) or a (feasible) reference trajectory, while providing stability and recursive feasibility guarantees for such settings~\cite{Mayne2000,rawlings2009model,borrelli2017predictive,Grune2011}.
	
	In practice, it is often difficult or cumbersome to pre-compute reference trajectories which are feasible, i.e., which both satisfy path constraints and evolve according to the system dynamics. It is therefore appealing to use reference trajectories (or reference paths) that are simple to compute and only satisfy path constraints, but not necessarily the system dynamics. However, using such trajectories introduces difficulties in providing stability guarantees~\cite{rawlings2012fundamentals}.
		
	This paper aims at (partially) filling a gap between practice and theory that exists for MPC formulations with infeasible references. Previous work~\cite{Rawlings2008a} has considered infeasible set points, and shown that stabilization to the closest feasible set point can be guaranteed. However, in the case of time-varying references, this analysis does directly not apply. To that end, we propose to use an Input-to-State Stability~(ISS) approach instead. 
	
	Our contribution in this paper is twofold. First, we prove that MPC formulations that are formulated with an infeasible reference, can actually stabilize towards an optimal trajectory, subject to specific terminal conditions. However, selecting such terminal conditions requires in general that an optimization problem is solved beforehand to compute a feasible reference. Consequently, if this step is performed, one has no reason to provide an infeasible reference to the MPC controller. Therefore, in a second step, we extend this first result to construct an ISS analysis to further show that if sub-optimal terminal conditions are chosen, the controlled system in closed-loop is stabilized around a neighborhood of the optimal trajectory.
		
	This paper is structured as follows. In Section~\ref{sec:mpftc} we introduce the MPC tracking problem, while in Section~\ref{sec:economic_mpc} we prove that, while having an ideal setting in mind,  one can design an MPC formulation that stabilizes the closed-loop system to an optimal feasible reference, even if an infeasible reference is used, at the price of defining suitable terminal conditions. Then, in Section~\ref{sec:iss} we further extend the results by proving ISS for practical settings, i.e., in case the terminal conditions are based on the infeasible reference. Finally, in Section~\ref{sec:simulations} we illustrate the derived theory with a numerical example, and draw conclusions in Section~\ref{sec:conclusions}.

	\section{Preliminaries}\label{sec:mpftc}
	Consider a discrete-time Linear Time-Varying~(LTV) system
	\begin{equation}\label{eq:sys}
		\x_{k+1}=f_k(\x_k,\u_k)=A_k \x_k + B_k \u_k,
	\end{equation}
	where $\x_k\in\mathbb{R}^{n_\x}$ and $\u_k\in\mathbb{R}^{n_\u}$ are the state and input vectors at time $k$, and the matrices $A_k\in\mathbb{R}^{n_\x\times n_\x}$ and $B_k\in\mathbb{R}^{n_\x\times n_\u}$ are time-varying. While we only consider LTV systems in this paper, we comment in Remark~\ref{rem:nl_sys} on how the results could be extended to general nonlinear systems. The state and inputs are subject to constraints $h(\x,\u):\mathbb{R}^{n_\x}\times\mathbb{R}^{n_\u}\rightarrow\mathbb{R}^{n_h}$, where the inequality $h(\x,\u)\leq 0$ is defined element-wise. The constraint $h(\x,\u)$ models, e.g., regions of the state space which should be avoided, and actuator limitations. Our objective is to control the system such that the state $\x_k$ tracks a user-provided parameterized reference trajectory $\r(t)=(\r^\x(t),\r^\u(t))$ as closely as possible. We assume that the reference trajectory is parameterized with time parameter $t$, with natural dynamics
	\begin{equation}
	\label{eq:tau_controlled}
	t_{k+1} = t_k + t_\mathrm{s},
	\end{equation}
	where $t_\mathrm{s}=1$ for discrete-time systems, or the sampling time for sampled-data systems. Throughout the remainder of the paper, we will refer to any time dependence of the reference using notation $(\r^\x_k,\r^\u_k):=(\r^\x(t_k),\r^\u(t_k))$.
	
	In order to track the reference $(\r^\x_k,\r^\u_k)$, we formulate the tracking MPC problem as 
		\begin{subequations}
		\label{eq:nmpc}
		\begin{align}
			\begin{split}\hspace{-0.5em}V(\x_k,t_k):=&\min_{\substack{\bx},\substack{\bu}} \sum_{n=k}^{k+N-1}
				q_\r(\xb,\ub,t_n)\\
				&\hspace{3.5em}+p_\r(\xb[k+N],t_{k+N})\hspace{-2em}
			\end{split}\label{eq:nmpc_cost}\\
			\text{s.t.}\ &\xb[k][k] = \x_{k},\label{eq:nmpcState} \\
			&\xb[n+1] = f_n(\xb,\ub),\label{eq:nmpcDynamics} & \hspace{-1em}n\in \mathbb{I}_k^{k+N-1},\\
			&h(\xb,\ub) \leq{} 0, \label{eq:nmpcInequality_known}& \hspace{-1em}n\in \mathbb{I}_k^{k+N-1},\\
			&\xb[k+N] \in\mathcal{X}^\mathrm{f}_\r(t_{k+N})\label{eq:nmpcTerminal},
		\end{align}
	\end{subequations}
	where $k$ is the current time instance and $N$ is the prediction horizon. In tracking MPC, typical choices for the stage and terminal costs are
	\begin{align}
	q_\r(\xb,\ub,t_n) &:= \matr{c}{\xb-\rx_n\\\ub-\ru_n}^\top{}\hspace{-0.7em}W\matr{c}{\xb-\rx_n\\\ub-\ru_n},\label{eq:stage_cost}\\
	p_\r(\xb,t_n) &:= (\xb-\rx_{n})^\top{}P(\xb-\rx_{n}),\label{eq:terminal_cost}
	\end{align}
	where $W\in\mathbb{R}^{(n_\x+n_\u) \times (n_\x+n_\u)}$ and $P\in\mathbb{R}^{n_\x\times n_\x}$ are symmetric positive-definite matrices. In order to avoid further technicalities, we avoid more general costs for the sake of simplicity. The predicted states and controls at the prediction time $n$ given the states at the current time $k$, are defined as $\xb$, and $\ub$, respectively. 
	The initial condition is enforced by constraint \eqref{eq:nmpcState}, and constraint \eqref{eq:nmpcDynamics} enforces the system dynamics. Constraint \eqref{eq:nmpcInequality_known} denotes constraints, e.g., state and control limits, and constraint \eqref{eq:nmpcTerminal} defines a terminal set containing the reference $\r$. Note that, differently from standard formulations, the terminal constraint depends on the time parameter $t_{k+N}$ relative to the reference.
	
	In the following, we first recall the standard stability properties of tracking MPC. Then, in Sections~\ref{sec:economic_mpc} and~\ref{sec:iss} we will derive  
	Input-to-State Stability~(ISS) results when the parameterized reference trajectory is not feasible with respect to the system dynamics.

	In order to prove stability, we introduce the following standard assumptions, see, e.g.,~\cite{rawlings2009model,Grune2011}. 
	\begin{Assumption}[System and cost regularity]\label{a:cont}
		
			The system model $f$ is continuous, and the stage cost $q_\r:\mathbb{R}^{n_\x}\times\mathbb{R}^{n_\u}\times\mathbb{R}\rightarrow\mathbb{R}_{\geq{}0}$, and terminal cost $p_\r:\mathbb{R}^{n_\x}\times\mathbb{R}\rightarrow\mathbb{R}_{\geq{}0}$, are continuous at the origin and satisfy $q_\r(\rx_k,\ru_k,t_k)=0$, and $p_\r(\rx_k,t_k)=0$. Additionally, $q_\r(\bx_k,\bu_k,t_k)\geq{}\alpha_1(\|\bx_k-\rx_k\|)$ for all feasible $\x_k$, $\u_k$, and  $p_\r(\bx_k,t_k)\leq\alpha_2(\|\bx_k-\rx_k\|)$, where $\alpha_1$ and $\alpha_2$ are $\mathcal{K}_\infty$-functions.
		
	\end{Assumption}
	
	\begin{Assumption}[Reference feasibility] \label{a:rec_ref}

			The reference trajectory satisfies the system dynamics~\eqref{eq:nmpcDynamics} and the system constraints~\eqref{eq:nmpcInequality_known}, i.e., $\r^\x_{k+1}=f_k(\r^\x_k,\r^\u_k)$ and $h(\r^\x_k,\r^\u_k) \leq{} 0$, $\forall{}k\in\mathbb{I}_0^\infty$.
			
	\end{Assumption}
	\begin{Assumption}[Stabilizing Terminal Conditions] \label{a:terminal}
		There exists a parametric stabilizing terminal set  $\mathcal{X}^\mathrm{f}_\r(t)$ and a terminal control law $\kappa^\mathrm{f}_\r(\x,t)$ yielding:
		\begin{align*}
		\mathbf{x}_{+}^\kappa=f_k(\mathbf{x}_k,\kappa^\mathrm{f}_\r(\x_k,t)), && t_+ = t_k + t_\mathrm{s},
		\end{align*}
		such that
		$p_\r(\x_{+}^\kappa,t_{+}) - p_\r(\x_k,t_k) \leq{} - q_\r(\x_k,\kappa^\mathrm{f}_\r(\x_k,t_k),t_k)$,
		$\x_k\in\mathcal{X}^\mathrm{f}_\r(t_k)\Rightarrow \x^\kappa_{+}\in\mathcal{X}^\mathrm{f}_\r(t_{+})$, and $h(\x_k,\kappa^\mathrm{f}_\r(\x_k,t_k)) \leq{} 0$ hold for all $k\in\mathbb{I}_0^\infty$.

	\end{Assumption}	

\begin{Proposition} [Nominal Asymptotic Stability]\label{prop:stab_feas}
	{Suppose that Assumptions \ref{a:cont}, \ref{a:rec_ref}, and \ref{a:terminal}  hold, 
		and that the initial state $(\x_k,t_k)$ at time $k$ belongs to the feasible set of Problem \eqref{eq:nmpc}. Then the system \eqref{eq:sys} in closed-loop with the solution of~\eqref{eq:nmpc} applied in receding horizon is an asymptotically stable system.} \label{prop:stable}
	\begin{proof}
		See the standard proof in, e.g., \cite{rawlings2009model,borrelli2017predictive}.
	\end{proof}
\end{Proposition}

	Proposition~\ref{prop:stab_feas} recalls the known stability results from the existing literature, which apply to tracking MPC schemes. The resulting design procedure for asymptotically stable tracking MPC is indeed complicated by the task of precomputing a feasible reference trajectory $(\r^\x_k,\r^\u_k)$ that satisfies Assumption~\ref{a:rec_ref}. However, in practice, it may be convenient to use a reference trajectory that is infeasible w.r.t. the system dynamics, yet simpler to define. While in standard MPC settings the stability with respect to an unreachable set point has been studied in~\cite{Rawlings2008a}, the approach therein applies to time-invariant infeasible references. In order to overcome such a limitation, we consider a setting where the reference can be time-varying and does not need to satisfy Assumption~\ref{a:rec_ref}, and the terminal conditions \eqref{eq:nmpcTerminal} do not need to hold at the reference trajectory, but in a neighborhood. While the results proposed in this paper are developed for a standard MPC formulation, we point out that they hold in other settings as well, including Model Predictive path Following Control~(MPFC)~\cite{Faulwasser2016} or Model Predictive Flexible trajectory Tracking Control~(MPFTC)~\cite{batkovic2020safe}.

	\section{Optimal Feasible Reference}
	Consider the optimal state and input trajectories obtained as the solution of the optimal control problem (OCP)
	\begin{subequations}
		\label{eq:ocp}
		\begin{align}\begin{split}
		\hspace{-1em}(\bx^\rr,\bu^\rr)\hspace{-.25em}:=\hspace{-.5em}\lim_{M\rightarrow\infty}\hspace{-0.3em}
		\arg&\min_{\bxi,\bet} \sum_{n=0}^{M-1}
		\hspace{-.3em}q_\r(\bxi_n,\bet_n,t_n)\hspace{-0.1em}+\hspace{-0.1em}p_\r(\bxi_M,t_M)\label{eq:ocp_cost} \hspace{-20em}\end{split}\\
		\text{s.t.}\ &\bxi_0=\x_{0}, \label{eq:ocpState} &\\
		&\bxi_{n+1} = f_n(\bxi_n,\bet_n),\label{eq:ocpDynamics} & n\in \mathbb{I}_{0}^{M-1},\\
		&h(\bxi_n,\bet_n) \leq{} 0, \label{eq:ocpInequality_known}& \hspace{-1em}n\in \mathbb{I}_0^{M-1},
		\end{align}
	\end{subequations}
	with the corresponding value function
	\begin{equation}
		V^\mathrm{O}(\x_k,t_k) := \lim_{M\rightarrow\infty}\sum_{n=0}^{M-1} q_\r(\x^\rr_n,\u_n^\rr,t_n)+p_\r(\x_M^\rr,t_M)
	\end{equation}
	The terminal cost in~\eqref{eq:ocp_cost} and initial state constraint~\eqref{eq:ocpState} can in principle be omitted or formulated otherwise, e.g., the terminal cost can be replaced with a terminal constraint instead,  but we include them in the formulation since they are often taking this form. We use here the same stage cost as in~\eqref{eq:nmpc_cost} and assume it is positive-definite. We exclude positive semi-definite costs solely for the sake of simplicity.
	We define the Lagrangian of the OCP~\eqref{eq:ocp} as
	\begin{align*}
	\mathcal{L}^\mathrm{O}(\bxi, \bet, \blambda,\bmu,\mathbf{t}) &= \blambda_0^\top (\bxi_0 - \x_{0}) +p_\r(\bxi_M,t_M)\\
	&+\lim_{M\rightarrow\infty}\sum_{n=0}^{M-1}
	q_\r(\bxi_n,\bet_n,t_n) +\bmu_n^\top h(\bxi_n,\bet_n)\\
	&+\lim_{M\rightarrow\infty}\sum_{n=0}^{M-1}  \blambda_{n+1}^\top (\bxi_{n+1} - f_n(\bxi_n,\bet_n)),
	\end{align*}
	and denote the optimal multipliers as $\blambda^\rr,\bmu^\rr$, and the solution of~\eqref{eq:ocp} as $\y^\rr:=(\x^\rr,\u^\rr)$. 
	Hereafter, we will refer to the reference $\y^\rr$ as the \emph{feasible reference}, as it satisfies Assumption~\ref{a:rec_ref}.
	
	\begin{Remark}
		Note that Problem~\eqref{eq:ocp} is formulated as an infinite horizon OCP since a reference could be defined over an infinite time horizon. For instance, a stationary reference can be viewed as being infinitely long as it remains at the same point at all times.
	\end{Remark}
	
	In the following, we will prove the stability of \eqref{eq:sys} w.r.t. $\y^\rr$ by relying on the trajectories~$\y^\rr$ and $\blambda^\rr$ from~\eqref{eq:ocp}, where $\y^\rr$ is used as an auxiliary reference.
	Our analysis will proceed as follows. We will first discuss an ideal case in which the terminal conditions are constructed based on $\y^{\mathrm{r}}$. By exploiting ideas from economic MPC we will prove that asymptotic stability can be obtained in that case. Since our objective is to avoid using any information on $\y^{\mathrm{r}}$, we will then turn to the realistic MPC formulation~\eqref{eq:nmpc}, and we will prove ISS.

	\subsection{Ideal MPC and Asymptotic Stability}	\label{sec:economic_mpc}
	
	Our analysis builds on tools that are used in the stability analysis of economic MPC schemes. The interested reader is referred to the following most relevant publications related to our analysis~\cite{Diehl2011,Amrit2011a,Zanon2018a,Faulwasser2018}. 
	
	Economic and tracking MPC schemes differ in the cost function, which satisfies
	\begin{align}
		\begin{split}
		\label{eq:tracking_cost}
		q_\r(\x^\rr_k,\u^\rr_k,t_k) =0,\ &q_\r(\x_k,\u_k,t_k) >0,\\
		&\forall \ \x_k\neq{}\bx^{\mathrm{r}}_k,\ \u_k\neq\bu^{\mathrm{r}}_k,
		\end{split}
	\end{align}
		in tracking schemes but not in economic ones. Note that~\eqref{eq:tracking_cost} can only hold if $\r=\y^\rr$, that is, if Assumption~\ref{a:rec_ref} holds.
	Consequently, even if the cost is positive-definite, any MPC scheme formulated with an infeasible reference $\r$ is an economic MPC. 
	We refer to~\cite{Zanon2018a,Faulwasser2018} for a detailed discussion on the topic.
	On the contrary, if $\y^\rr$ is used as reference, we obtain the tracking stage cost $q_{\y^\rr}$. Since precomputing a feasible reference $\y^\rr$ can be impractical or involved, we focus next on the case of \emph{infeasible references}.

	Consider the Lagrangian of the OCP~\eqref{eq:ocp}
	\begin{align*}
	\mathcal{L}^\mathrm{O}(\bxi, \bet, \blambda,\bmu,\mathbf{t}) &= \blambda_0^\top (\bxi_0 - \x_{0}) +p_\r(\bxi_M,t_M)\\
	&+\lim_{M\rightarrow\infty}\sum_{n=0}^{M-1}
	q_\r(\bxi_n,\bet_n,t_n) +\bmu_n^\top h(\bxi_n,\bet_n)\\
	&+\lim_{M\rightarrow\infty}\sum_{n=0}^{M-1}  \blambda_{n+1}^\top (\bxi_{n+1} - f_n(\bxi_n,\bet_n)),
	\end{align*}
	and denote the optimal multipliers as $\blambda^\mathrm{r},\bmu^\mathrm{r}$, and the solution of~\eqref{eq:ocp} as $\y^\rr:=(\x^\rr,\u^\rr)$.	In order to construct a tracking cost from the economic one, we use the Lagrange multipliers of the OCP~\eqref{eq:ocp} to construct a \emph{rotated} problem, which has the same constraints as the original MPC problem~\eqref{eq:nmpc} and the following \emph{rotated stage and terminal costs} 
		\begin{align*}
	&\bar q_\r(\xb,\ub,t_n):=q_\r(\xb,\ub,t_n)-q_\r(\x^\rr_n,\u^\rr_n,t_n)\\
	&\hspace{1em}+ \blambda_n^{\mathrm{r}\top}(\xb[n][k]-\x^\rr_n)- \blambda^{\rr\top}_{n+1} (f_n(\xb[n][k],\ub[n][k])-f_n(\x^\rr_n,\u^\rr_n)), \\
	&\bar{p}_\r(\xb,t_n):= p_\r(\xb)-p_\r(\x_{n}^\rr,t_n)+\blambda^{\rr\top}_{n}(\xb-\x^\rr_{n}).
	\end{align*}

As we prove in the following Lemma~\ref{lem:rot_ocp}, adopting the rotated stage cost $\bar q_\r$ and terminal cost $\bar p_\r$ in the OCP~\eqref{eq:ocp} does not change its primal solution. Such property of the rotated costs will be exploited next in the formulation of the \emph{ideal} MPC problem.
	\begin{Lemma}
		\label{lem:rot_ocp}
		If OCP~\eqref{eq:ocp} is formulated using the rotated cost instead of the original one, then the Second Order Sufficient optimality Conditions (SOSC) are satisfied~\cite{Nocedal2006}, and the following claims hold:
		\begin{enumerate}
		    \item[i)] the primal solution is unchanged;
		    \item[ii)] the rotated cost penalizes deviations from the optimal solution of Problem~\eqref{eq:ocp}, i.e.,
		    \begin{align*}
		        \bar q_\r(\x_n^\rr,\u_n^\rr,t_n) =0,\ \bar q_\r(\x_n,\u_n,t_n)>0,
		    \end{align*}
		    for all $(\x_n,\u_n) \neq (\bx_n^\rr,\bu_n^\rr)$ satisfying $h(\x_n,\u_n) \leq 0$.
		\end{enumerate}
	\end{Lemma}
	\begin{proof}
		First, we prove that if Problem~\eqref{eq:ocp} is formulated using stage cost $\bar q_\r$ and terminal cost $\bar p_\r$ instead of $q_\r$ and $p_\r$, the primal solution remains unchanged. 
		This is a known result from the literature on economic MPC and is based on the observation that all terms involving $\blambda^\mathrm{r}$ in the rotated cost form a telescopic sum and cancel out, such that only ${\blambda_0^\mathrm{r}}^\top (\bxi_0-\x_0^\mathrm{r})$ remains. Since the initial state is fixed, the cost only differs by a constant term and the primal solution is unchanged. The cost $\bar q_\r$ being nonnegative is a consequence of the fact that the stage cost Hessian is positive definite by Assumption \ref{a:cont}, the system dynamics are LTV, and the Lagrange multipliers $\bar \blambda$ associated with Problem~\eqref{eq:ocp} using cost $\bar q_\r$ are $0$. 
		
		To prove the second claim, we define the Lagrangian of the rotated problem as 
		\begin{align*}
		\mathcal{\bar L}^\mathrm{O}(\bxi, \bet, \bar \blambda,\bar \bmu,\mathbf{t}) 
		= \ & \bar{\blambda}_0^\top (\bxi_0 - \x_{0}) + \bar p_\r (\bxi_M,t_M)\\
		&\hspace{-2em}+\lim_{M\rightarrow\infty}\sum_{n=0}^{M-1}
		\bar{q}_\mathrm{\r}(\bxi_n,\bet_n,t_n) + \bar \bmu_n^\top h(\bxi_n,\bet_n)\\
		&\hspace{-2em}+\lim_{M\rightarrow\infty}\sum_{n=0}^{M-1} \bar{\blambda}_{n+1}^\top ( \bxi_{n+1} - f_n(\bxi_n,\bet_n) ).
		\end{align*}
		For compactness we denote next $\nabla_n:=\nabla_{(\bxi_n,\bet_n)}$. Since by construction $\nabla_n \bar q_\mathrm{\r}=\nabla_n \mathcal{L}^\mathrm{O} - \nabla_n \bmu_n^{\rr\top} h $, we obtain
		\begin{align*}
		\nabla_n \mathcal{\bar L}^\mathrm{O} &= \nabla_n \bar q_\mathrm{\r} + \matr{c}{\bar \blambda_n \\ 0} - \nabla_n \bar \blambda_{n+1}^\top f_n  + \nabla_n \bar \bmu_n^\top h \\
		&\hspace{-1.2em}= \nabla_n \mathcal{L}^\mathrm{O} + \matr{c}{\bar \blambda_{n} \\ 0} - \nabla_n \bar \blambda_{n+1}^\top f_n + \nabla_n (\bar \bmu_n-\bmu_n^\mathrm{r})^\top h.
		\end{align*}
		Therefore, the KKT conditions of the rotated problem are solved by the same primal variables as the original problem and $\bar \bmu_n = \bmu_n^\mathrm{r}$, $\bar \blambda_n=0$. With similar steps we show that $\bar{\lambda}_P=0$, since $\nabla_M\bar{p}_\r=\nabla_M \mathcal{L}^\mathrm{O}$.
		Because the system dynamics are LTV and the stage cost quadratic, we have that 
		$\nabla^2_n \bar q_\mathrm{\r} = \nabla^2_n q_\mathrm{\r}\succ0$
		Moreover, since the solution satisfies the SOSC,
		we directly have that $\bar q_\mathrm{\r}(\x_n^\mathrm{r},\u_n^\mathrm{r},t_n) =0$ and $\bar q_\mathrm{\r}(\bxi_n,\bet_n,t_n) > 0$ for all $(\bxi_n,\bet_n)\neq(\x_n^\mathrm{r},\u_n^\mathrm{r})$ s.t. $h(\bxi_n,\bet_n) \leq 0$.
	\end{proof}
	\begin{Remark}
		\label{rem:nl_sys}
		The only reason that limits our result to LTV systems is that this entails $\nabla^2_n \bar q_\mathrm{\r} = \nabla^2_n q_\mathrm{\r}\succ0$. It seems plausible that this limitation could be overcome by assuming that OCP~\eqref{eq:ocp} satisfies the SOSC for all initial states at all times. However, because further technicalities would be necessary to obtain the proof, we leave this investigation for future research.
	\end{Remark}
	\begin{Corollary} The rotated value function of OCP~\eqref{eq:ocp}, i.e.,
	    \begin{align*}
		    \bar V^\mathrm{O}(\x_k,t_k) &=\ V^\mathrm{O}(\x_k,t_k) + \blambda^{\rr^\top}_k (\x_k-\x^\rr)\\
		    &-\lim_{M\rightarrow\infty}\sum_{n=k}^{k+M-1}q_\r(\x^\rr_n,\u^\rr_n,t_n)-p_\r(\x^\rr_{k+M},t_M),
	    \end{align*}
		is positive definite, and its minimum is $\bar V^\mathrm{O}(\bx_k^\mathrm{r},t_k)=0$.
	\end{Corollary}
	\begin{proof}
	We note from the proof in Lemma~1, that the rotated stage and terminal costs are positive definite and that they are zero at the feasible reference $(\x^\rr_n,\u^\rr_n)$, hence, the rotated value function is also positive, and zero at $\x^\r_k$.
	\end{proof}

	\paolor{While Proposition~\ref{prop:stab_feas} proves the stability of system~(1) in closed-loop with the solution of~\eqref{eq:nmpc} under Assumptions~\ref{a:rec_ref} and~\ref{a:terminal}, in Theorem~\ref{thm:as_stab_0} we  will prove stability in case the reference trajectory does not satisfy Assumption~\ref{a:rec_ref}. The stability proof in Theorem~\ref{thm:as_stab_0} 
builds on the following \emph{ideal} formulation} 
	\begin{subequations}
	\label{eq:ideal_nmpc}
	\begin{align}
	\begin{split}V^\mathrm{i}(\x_k,t_k) = \min_{\bx,\bu}&\sum_{n=k}^{k+N-1} q_\r(\xb,\ub,t_n) \\
		&\hspace{2em}+p_{\tilde\y^\rr}(\xb[k+N],t_{k+N})
	\end{split} \\
		\mathrm{s.t.} \ \ &\eqref{eq:nmpcState}-\eqref{eq:nmpcInequality_known}, \ \xb[k+N] \in\mathcal{X}^\mathrm{f}_{\y^\rr}(t_{k+N}),\label{eq:ideal_nmpc_terminal}
	\end{align}
\end{subequations}
where
\begin{align}\label{eq:minimizer_tilde_yr}
		\tilde\y^\rr_k &:= \arg\min_{\x} p_{\y^\rr}(\x,t_k)-\blambda_k^{\rr\top}(\x-\x^\rr_k).
\end{align}
The Problems~\eqref{eq:nmpc} and~\eqref{eq:ideal_nmpc} only differ in the terminal cost and constraint: in~\eqref{eq:ideal_nmpc} they are written with respect to the solution $\y^\rr$ and $\blambda^\rr$ of~\eqref{eq:ocp} rather than~$\r$. In order to distinguish the solutions of~\eqref{eq:nmpc} and~\eqref{eq:ideal_nmpc}, we denote the solution of~\eqref{eq:nmpc} by $\x^\star$, $\u^\star$, and the solution of~\eqref{eq:ideal_nmpc} by $\x^\mathrm{i}$, $\u^\mathrm{i}$. In addition, when the stage cost $\bar{q}_\r$ and terminal cost $\bar p_{\by^\rr}$ are used, we obtain the corresponding \emph{rotated} formulation of~\eqref{eq:ideal_nmpc}
\begin{align}
	\label{eq:ideal_rot_nmpc}
	\begin{split}\bar V^\mathrm{i}(\x_k,t_k) = \min_{\x,\u}  &\sum_{n=k}^{k+N-1} \bar q_\r(\xb,\ub,t_n) \\
		&\hspace{2em}+ \bar p_{\tilde\y^\mathrm{r}}(\xb[k+N],t_{k+N})
	\end{split} \\
	\mathrm{s.t.}\hspace{0em} \ \ &\eqref{eq:nmpcState}-\eqref{eq:nmpcInequality_known}, \ \xb[k+N] \in\mathcal{X}^\mathrm{f}_{\y^\mathrm{r}}(t_{k+N}),\nonumber
\end{align}
where the rotated terminal cost is defined as
\begin{align}\begin{split}\label{eq:rot_tilde_terminal_cost}
	\bar{p}_{\by^\rr}(\xb,t_n)&:= p_{\by^\rr}(\xb,t_n)-p_{\by^\rr}(\x_{n}^\rr,t_n)\\
	&+\blambda^{\rr\top}_{n}(\xb-\x^\rr_{n}).\end{split}
\end{align}
Note that by Lemma~\ref{lem:rot_ocp}, the rotated cost $\bar q_\r$ penalizes deviations from $\y^\rr$, i.e., the solution to \eqref{eq:ocp}. We will prove next that $\bar p_{\by^\r}$ also penalizes deviations from $\y^\r$, implying that \emph{the rotated ideal MPC formulation is of tracking type}.

	\begin{Lemma}
		\label{lem:rot_mpc}
		Consider the \emph{rotated} \emph{ideal} MPC Problem~\eqref{eq:ideal_rot_nmpc}, formulated using the rotated costs $\bar q_\r$ and $\bar{p}_{\by^\r}$, and the terminal set $\mathcal{X}_{\y^\r}^\mathrm{f}$. Then, the primal solution of~\eqref{eq:ideal_rot_nmpc} coincides with the primal solution of the ideal MPC Problem~\eqref{eq:ideal_nmpc}.
	\end{Lemma}
	\begin{proof}
		From~\eqref{eq:minimizer_tilde_yr} and~\eqref{eq:rot_tilde_terminal_cost} we have that $\bar p_{\by^\r}(\x_k^\rr,t_k) =0$ and that
		$\nabla \bar p_{\by^\rr}(\x^\rr_k,t_k) = \nabla p_{\by^\rr}(\x^\rr_k,t_k) + \nabla p_{\y^\rr}(\by^\rr_k,t_k) = 0$, since the terminal costs are quadratic~\eqref{eq:terminal_cost}. The proof then follows along the same lines as Lemma~\ref{lem:rot_ocp} and~\cite{Diehl2011,Amrit2011a}.
	\end{proof}

	In order to prove Theorem~\ref{thm:as_stab_0}, we  need that the terminal conditions of the rotated ideal formulation~\eqref{eq:ideal_rot_nmpc} satisfy Assumption~\ref{a:terminal}. To that end, we introduce the following assumption.
	\begin{Assumption}\label{a:terminal_for_rotated}
	There exists a parametric stabilizing terminal set  $\mathcal{X}^\mathrm{f}_{\y^\rr}(t)$ and a terminal control law $\kappa^\mathrm{f}_{\y^\rr}(\x,t)$ yielding:
	\begin{align*}
		\mathbf{x}_{+}^\kappa=f_k(\mathbf{x}_k,\kappa^\mathrm{f}_{\y^\rr}(\x_k,t)), && t_+ = t_k + t_\mathrm{s},
	\end{align*}
	so that
	$\bar p_{\by^\rr}(\x_{+}^\kappa,t_{+})- \bar p_{\by^\rr}(\x_k,t_k) \leq{} - \bar q_\r(\x_k,\kappa^\mathrm{f}_{\y^\rr}(\x_k,t_k),t_k)$, $\x_k\in\mathcal{X}^\mathrm{f}_{\y^\rr}(t_k)\Rightarrow \x^\kappa_{+}\in\mathcal{X}^\mathrm{f}_{\y^\rr}(t_{+})$, and $h(\x_k,\kappa^\mathrm{f}_{\y^\rr}(\x_k,t_k)) \leq{} 0$ hold for all $k\in\mathbb{I}_0^\infty$.
\end{Assumption}	
Note that Assumption~\ref{a:terminal_for_rotated} only differs from Assumption~\ref{a:terminal} by the fact that the set and control law are centered on $\y^\rr$ rather than $\r$, and that the costs are rotated.
	\begin{Theorem}
		\label{thm:as_stab_0}
		Suppose that Assumptions \ref{a:cont} and~\ref{a:terminal_for_rotated} hold, and that Problem~\eqref{eq:ocp} is feasible for initial state $(\x_k,t_k)$. Then, system~\eqref{eq:sys} in closed-loop with the ideal MPC~\eqref{eq:ideal_nmpc} is asymptotically stabilized to the optimal trajectory $\bx^{\mathrm{r}}$.
	\end{Theorem}
	\begin{proof}
		By Lemma~\ref{lem:rot_mpc}, the rotated ideal MPC problem has positive-definite stage and terminal costs penalizing deviations from the optimal trajectory $\y^{\mathrm{r}}$. Hence, the rotated ideal MPC problem is of tracking type. 
			
		Assumption~\ref{a:cont} directly entails a lower bounding by a $\mathcal{K}_\infty$ function, and can also be used to prove an upper bound~\cite[Theorem 2.19]{rawlings2009model}, such that the following holds
		\begin{equation*}
			\alpha_1(\|\x_k-\x^\rr_k\|) \leq{} \bar V^\mathrm{i}(\x_k,t_k)\leq{} \alpha_2(\|\x_k-\x^\rr_k\|),
		\end{equation*}
		where $\alpha_1,\alpha_2\in\mathcal{K}_\infty$. Then, solving Problem~\eqref{eq:ideal_rot_nmpc}, we obtain $\bar V^{\mathrm{i}}(\x_k,t_k)$ and the optimal trajectories $\{\xb[k]^\mathrm{i},...,\xb[k+N]^\mathrm{i}\}$, and $\{\ub[k]^\mathrm{i},...,\ub[k+N-1]^\mathrm{i}\}$. By relying on Assumptions~\ref{a:rec_ref} and~\ref{a:terminal}, using terminal control law $\kappa^\mathrm{f}_{\y^\rr}$, we can construct the feasible sub-optimal trajectories $\{\xb[k+1]^\mathrm{i},...,\xb[k+N]^\mathrm{i},f_{k+N}(\xb[k+N]^\mathrm{i},\kappa^\mathrm{f}_{\y^\r})\}$ and $\{\ub[k+1]^\mathrm{i},...,\ub[k+N]^\mathrm{i},\kappa^\mathrm{f}_{\y^\rr}\}$ at time $k+1$, which can be used to derive the decrease condition following standard arguments~\cite{rawlings2009model,borrelli2017predictive}:
		$$\bar{V}^\mathrm{i}(\x_{k+1},t_{k+1})-\bar{V}^\mathrm{i}(\x_k,t_k)\leq{}-\alpha_3(\|\x_k-\x^\rr_k\|).$$ 
		This entails that the \emph{rotated} \emph{ideal} value function $\bar{V}^\mathrm{i}(\x_k,t_k)$ is a Lyapunov function, and that the closed-loop system is asymptotically stabilized to $\x^\rr$. 
		Finally, using Lemma~\ref{lem:rot_mpc} we establish asymptotic stability also for the \emph{ideal} MPC scheme~\eqref{eq:ideal_nmpc}, since the primal solutions of the two problems coincide.
	\end{proof}

	Theorem~\ref{thm:as_stab_0} establishes the first step towards the desired result: 
	an MPC problem can be formulated using an \emph{infeasible reference}, which stabilizes system~\eqref{eq:sys} to the optimal trajectory of Problem~\eqref{eq:ocp} provided that the appropriate terminal conditions are used.

	At this stage, the main issue is to express the terminal constraint set as 
	a positive invariant set containing $\bx^{\mathrm{r}}$, and the terminal control law stabilizing the system to  $\bx^{\mathrm{r}}$.
	To that end, one needs to know the feasible reference trajectory~$\bx^{\mathrm{r}}$, i.e., to solve Problem~\eqref{eq:ocp}. Since solving Problem~\eqref{eq:ocp} is not practical, we prove in the next section how sub-optimal terminal conditions can be used instead to prove ISS for the closed-loop system.

	\subsection{Practical MPC and ISS}\label{sec:iss}
	
	In this subsection, we analyze the case in which the terminal conditions are not enforced based on the feasible reference trajectory, but 
	rather based on an \emph{approximatively feasible} reference (see Assumption~\ref{a:approx_feas}). 
	Since in that case asymptotic stability cannot be proven, we will prove ISS for the closed-loop system, where the input is some terminal reference $\y^{\mathrm{f}}$. In particular, we are interested in the practical approach $\y^{\mathrm{f}}=\r(t_{k+N})$ and the ideal setting $\y^{\mathrm{f}}=\y^{\mathrm{r}}(t_{k+N})$. 
	To that end, we define the following closed-loop dynamics
	\begin{align}\label{eq:iss_dynamics}
	\x_{k+1}(\y^\mathrm{f}) = f_k(\x_{k},\u_\mathrm{MPC}(\x_{k},\y^\mathrm{f})) = \bar f_k(\x_{k},\y^\mathrm{f}),
	\end{align}
	where we stress that~$\u_\mathrm{MPC}$ is obtained as~$\ub[k]^\star$ solving problem~\eqref{eq:nmpc} (in case one uses $\y^\mathrm{f}=\r$ and terminal cost $p_\r$); or as~$\ub[k]^\mathrm{i}$ solving the ideal problem~\eqref{eq:ideal_nmpc} (in case one uses $\y^\mathrm{f}=\y^\rr$ and terminal cost $p_{\tilde\y^\r}$). In the following we will use the notation $\x_{k+1}(\y^\mathrm{f})$ to stress that the terminal reference $\y^\mathrm{f}$ is used in the computation of the control yielding the next state. Additionally, we define 
	the following quantities 
	\begin{align*}
	\bar J_{\by^\mathrm{r}}^{\star}(\x_k,t_k) &:= \sum_{n=k}^{N-1} \bar q_\r(\xb^\star,\ub^\star,t_n) + \bar p_{\by^\mathrm{r}}(\xb[k+N]^\star,t_{k+N}), \\
	\bar J_\r^{\mathrm{i}}(\x_k,t_k) &:= \sum_{n=k}^{N-1} \bar q_\r(\xb^\mathrm{i},\ub^\mathrm{i},t_n) + \bar p_\r(\xb[k+N]^\mathrm{i},t_{k+N}),
	\end{align*} 
	and we remind that
	\begin{align*}
	\bar V(\x_k,t_k) &= \sum_{n=k}^{N-1} \bar q_\r(\xb^\star,\ub^\star,t_n) + \bar p_\r(\xb[k+N]^\star,t_{k+N}),\\
	\bar V^\mathrm{i}(\x_k,t_k) &= \sum_{n=k}^{N-1} \bar q_\r(\xb^\mathrm{i},\ub^\mathrm{i},t_n) + \bar p _{\by^\mathrm{r}}(\xb[k+N]^\mathrm{i},t_{k+N}).
	\end{align*}
	
	Before formulating the stability result in the next theorem, we need to introduce an additional assumption on the reference infeasibility.
	\begin{Assumption}[Approximate feasibility of the reference]
		\label{a:approx_feas}
		The reference $\y^{\mathrm{f}}$ satisfies the  constraints \eqref{eq:nmpcInequality_known}, i.e., $h(\bx^{\mathrm{f}}_n,\bu^{\mathrm{f}}_n) \leq{} 0$, $n\in \mathbb{I}_k^{k+N-1}$, for all $k\in\mathcal{N}^+$. Additionally, recursive feasibility holds for both Problem~\eqref{eq:nmpc} and~\eqref{eq:ideal_nmpc} when the system is controlled in closed-loop using the feedback from Problem~\eqref{eq:nmpc}.
	\end{Assumption}
	
	\begin{Remark}
		Assumption~\ref{a:approx_feas} essentially only requires that the reference used in the definition of the terminal conditions (constraint and cost) is feasible with respect to the system constraints, and not the system dynamics. However, recursive feasibility holds if the reference satisfies, e.g., $\|\x_{n+1}^\mathrm{f}-f_n(\x_n^\mathrm{f},\u_n^\mathrm{f})\|\leq{}\epsilon$, for some small $\epsilon$ i.e., if the reference satisfies the system dynamics approximately.
		Note that, if $\epsilon=0$, then Assumption~\ref{a:rec_ref} is satisfied and Assumption~\ref{a:approx_feas} is not needed anymore. Finally, the infeasibility due to $\epsilon\neq0$ could be formally accounted for so as to satisfy Assumption~\ref{a:approx_feas} by taking a robust MPC approach, see, e.g.,~\cite{Mayne2005,Chisci2001}. 
	\end{Remark}
From a practical standpoint, Assumption~\ref{a:approx_feas} sets a rather mild requirement. In fact, it is not uncommon to use infeasible references for simplicity or satisfying approximated system dynamics to capture the most relevant dynamics of the system (keeping $\epsilon$ small).
	
	{We are now ready to state the main result of the paper.}

	\begin{Theorem}\label{thm:iss}
		Suppose that Problem~\eqref{eq:ocp} is feasible and Assumptions~\ref{a:cont} and~\ref{a:terminal} hold for the reference $\y^{\mathrm{r}}$ with costs $\bar{q}_\r$ and $\bar{p}_{\by^\r}$ and terminal set $\mathcal{X}_{\y^\r}$. Suppose moreover that Problem~\eqref{eq:nmpc} and Problem~\eqref{eq:ideal_nmpc} are feasible at time $k$ with inital state $(\x_k,t_k)$, and that reference $\y^\mathrm{f}$, with terminal set $\mathcal{X}^\mathrm{f}_{\y^\mathrm{f}}$, satisfies Assumption~\ref{a:approx_feas}. Then, system~\eqref{eq:iss_dynamics} obtained from~\eqref{eq:sys} in closed-loop with MPC formulation~\eqref{eq:nmpc} is ISS.
	\end{Theorem}
	\begin{proof}
		
		We prove the result using the value function $\bar V^\mathrm{i}(\x_k,t_k)$ of the rotated ideal Problem~\eqref{eq:ideal_rot_nmpc} as an ISS-Lyapunov function candidate \cite{jiang2001input}. From the prior analysis in Theorem \ref{thm:as_stab_0} we know that Assumption~\ref{a:rec_ref} holds for $\y^\rr$ since Problem~\eqref{eq:ocp} is feasible, and that $\bar V^\mathrm{i}(\x_k,t_k)$ is a Lyapunov function {when the \emph{ideal} terminal conditions} $\y^\mathrm{f}=\y^\rr$ are used. Hence, when {we apply the ideal control input $\ub[k][k]^\mathrm{i}$, i.e., use \eqref{eq:iss_dynamics} to obtain the next state $\x_{k+1}(\y^\rr)=\bar{f}_k(\x_k,\y^\rr)$}, we have the following relations
		\begin{align*}
			\alpha_1(\| \x_k-\x^\rr_k \|) \leq \bar V^\mathrm{i}(\x_k,t_k) \leq \alpha_2(\| \x_k-\x^\rr_k \|),\\
			\bar V^\mathrm{i}(\x_{k+1}(\y^\rr),t_{k+1}) - \bar V^\mathrm{i}(\x_{k},t_k) \leq -\alpha_3(\| \x_k-\x^\rr_k \|),
		\end{align*}
		with $\alpha_i\in \mathcal{K}_\infty$, $i=1,2,3$. 
		
		We are left with proving ISS, i.e., that $\exists \, \sigma\in\mathcal{K}$ such that when the reference $\y^\mathrm{f}$ is treated as an external input, the next state is given by $\x_{k+1}(\y^\mathrm{f})=\bar f_k(\x_k,\y^\mathrm{f})$, the following holds
		\begin{align}\begin{split}
			\label{eq:iss_decrease}
			\bar V^\mathrm{i}({\x_{k+1}(\y^\mathrm{f})},t_{k+1})- \bar V^\mathrm{i}(\x_{k},t_k)\leq&\sigma( \|  \y^\mathrm{f}-\y^\rr  \| )\\&-\alpha_3(\| \x_k-\bx^\rr_k \|).
		\end{split}\end{align}
		
		In order to bound $\bar V^\mathrm{i}({\x_{k+1}(\y^\mathrm{f})},t_{k+1}) - \bar V^\mathrm{i}(\x_{k},t_{k})$, we first derive an upper bound on $\bar J_\r^\mathrm{i}$ which depends on $\bar V^\mathrm{i}$.
		To that end, we observe that the rotated cost of the ideal trajectory $\xb^\mathrm{i}$, $\ub^\mathrm{i}$ satisfies 
		\begin{align*}
			\bar J_\r^\mathrm{i}(\x_{k},t_k)&=  \bar V^\mathrm{i}(\x_{k},t_k)-\bar p_{\by^\mathrm{r}}(\xb[k+N]^\mathrm{i},t_{k+N})\\
			&+\bar p_\r(\xb[k+N]^\mathrm{i},t_{k+N}).
		\end{align*}
		Defining
		\begin{align*}
			\phi(\y^\mathrm{f})&:=\bar p_{\y^\mathrm{f}}(\xb[k+N]^\mathrm{i},t_{k+N})-\bar p_{\by^\mathrm{r}}(\xb[k+N]^\mathrm{i},t_{k+N}),
		\end{align*} 
		there exists a $\sigma_1 \in \mathcal{K}$ such that $\phi(\y^\mathrm{f}) \leq{} \sigma_1(\|\y^\mathrm{f}-\y^\r\|)$
		since, by~\eqref{eq:terminal_cost}, $\phi(\y^\mathrm{f})$ is a continuous function of $\y^\mathrm{f}$ and $\phi(\y^\mathrm{r})=0$.
		Then, the following upper bound is obtained
		\begin{align*}
			\bar J_\r^\mathrm{i}(\x_{k},t_k)&\leq \bar V^\mathrm{i}(\x_{k},t_k) + \sigma_1(\| \y^\mathrm{f}-\y^\mathrm{r} \| ).
		\end{align*}
		Upon solving Problem~\eqref{eq:nmpc}, we obtain $\bar V(\x_{k},t_k)\leq\bar J_\r^\mathrm{i}(\x_{k},t_k)$. Starting from the optimal solution $\x^\star$, and $\u^\star$, we will construct an upper bound on the decrease condition. To that end, we first need to evaluate the cost of this trajectory, i.e., 
		\begin{align*}\begin{split}
			\bar J_{\by^\mathrm{r}}^{\star}(\x_{k},t_k)&=\bar V(\x_{k},t_k)-\bar p_\r(\xb[k+N]^\star,t_{k+N})\\
			&+\bar p_{\by^\mathrm{r}}(\xb[k+N]^\star,t_{k+N}).
		\end{split}\end{align*}
		Using the same reasoning as before, there exists $\sigma_2 \in \mathcal{K}$ such that
		\begin{align*}
			&\bar p_{\by^\mathrm{r}}(\xb[k+N]^\star,t_{k+N})-\bar p_\r(\xb[k+N]^\star,t_{k+N})\\
			&\hspace{5em}\leq \sigma_2(\| \y^\mathrm{f}_{k+N}-\y^\mathrm{r}_{k+N} \| ).
		\end{align*}
		Then, we obtain
		\begin{align}\label{eq:jbar}
			\begin{split}
			\bar J_{\by^\mathrm{r}}^{\star}(\x_{k},t_k) &\leq \bar V(\x_{k},t_k) + \sigma_2(\| \y^\mathrm{f}_{k+N}-\y_{k+N}^\rr \| ) \\
			&\leq \bar J_\r^\mathrm{i}(\x_{k},t_k) + \sigma_2(\| \y^\mathrm{f}_{k+N}-\y_{k+N}^\rr \| ) \\
			& \leq \bar V^\mathrm{i}(\x_{k},t_k) + \sigma(\| \y^\mathrm{f}_{k+N}-\y_{k+N}^\rr \| ),
			\end{split}
		\end{align}
		where we defined $\sigma:=\sigma_1+\sigma_2$. 
		
	Proceeding similarly as in the proof of Proposition~\ref{prop:stable}, we apply the control input $\ub[k]^\star$ from~\eqref{eq:nmpc}, i.e., $\y^\mathrm{f}=\r$, to obtain $$\x_{k+1}(\y^\mathrm{f})=\bar{f}_k(\x_k,\y^\mathrm{f}).$$
		In order to be able to apply this procedure, we first assume that the obtained initial guess is feasible for the ideal problem~\eqref{eq:ideal_nmpc} and proceed as follows. 
		By Assumption~3, we use the terminal control law {$\kappa_{\y^\r}^\mathrm{f}(\x,t)$}, to form a guess at the next time step and upper bound the \emph{ideal} rotated value function. By optimality
		\begin{align}\label{eq:iss_value_decrease}
			\bar{V}^\mathrm{i}&( \x_{k+1}(\y^\mathrm{f}),t_{k+1}) \leq{}\\ &\nonumber\leq{}\sum_{n=k+1}^{N-1}\bar{q}_\r(\xb^\star,\ub^\star,t_n)+\bar q_\r(\xb[k+N]^\star,\kappa_{\y^\r},t_{k+N})\\
			&\nonumber\hspace{4em}+\bar p_{\by^\r}(\xb[k+N+1]^{\kappa,\star},t_{k+N+1})\\
			&\nonumber=\bar{J}_{\by^\r}^\star(\x_k,t_k)-\bar{q}_\r(\xb[k]^\star,\ub[k]^\star,t_k)-\bar{p}_{\by^\r}(\xb[k+N]^\star,t_{k+N})\\
			&\nonumber+\bar{p}_{\by^\r}(\xb[k+N+1]^{\star,\kappa},t_{k+N+1})+\bar{q}_{\r}(\xb[k+N]^\star,\kappa_{\y^\r},t_{k+N}),
		\end{align}
		where we used
		$$\xb[k+N+1]^{\star,\kappa}\hspace{-0.2em}:= \hspace{-0.2em}f_{k+N}(\xb[k+N],\kappa_{\y^\r}),\, \kappa_{\y^\r}\hspace{-0.2em}:=\hspace{-0.2em}\kappa_{\y^\r}(\xb[k+N]^\star,t_{k+N}),$$
		{and assumed that $\xb[k+N+1]^{\star,\kappa}\in\mathcal{X}_{\y^\r}(t_{k+N+1})$}. Again, using Assumption~3 we can now upper bound the terms
		\begin{align*} \bar{p}_{\by^\r}(\xb[k+N+1]^{\star,\kappa},t_{k+N+1})-\bar{p}_{\by^\r}(\xb[k+N]^\star,t_{k+N})\\\
		+\bar{q}_\r(\xb[k+N]^\star,\kappa_{\y^\r})\leq{}0,
		\end{align*}
		so that~\eqref{eq:iss_value_decrease} can be written as
		\begin{align}
			\bar{V}^\mathrm{i}(\x_{k+1}(\y^\mathrm{f}),t_{k+1}) &\leq{}J_{\by^\r}^\star(\x_k,t_k)-\bar{q}_\r(\xb[k]^\star,\ub[k]^\star,t_{k}),\\
			\bar{V}^\mathrm{i}(\x_{k+1}(\y^\mathrm{f})) &\leq{}J_{\by^\r}^\star(\x_k,t_k)-\alpha_3(\|\x_k-\x^\r_k\|),\label{eq:iss_bound_decr}
		\end{align}
		which, in turn, proves~\eqref{eq:iss_decrease}.

		In case ${\xb[k+N+1]^{\star,\kappa}\not\in\mathcal{X}^\mathrm{f}_{\y^\mathrm{r}}(t_{k+N+1})}$, 
		we resort to a relaxation of the terminal constraint with an exact penalty~\cite{Scokaert1999a,Fletcher1987} in order to compute an upper bound to the cost. This relaxation has the property that  the solution of the relaxed formulation coincides with the one of the non-relaxed formulation whenever it exists. Then, by construction, the cost of an infeasible trajectory is higher than that of the feasible solution.  
		
		{Finally, from Assumption~\ref{a:approx_feas} we know that the value functions $\bar{V}(\x_{k+1}(\y^\mathrm{f}),t_{k+1})$ and $\bar V^\mathrm{i}(\x_{k+1}(\y^\mathrm{f}),t_{k+1})$ are feasible and bounded for time $k+1$.}
	\end{proof}

\begin{figure*}[ht]
	\centering
	\includegraphics[width=\linewidth]{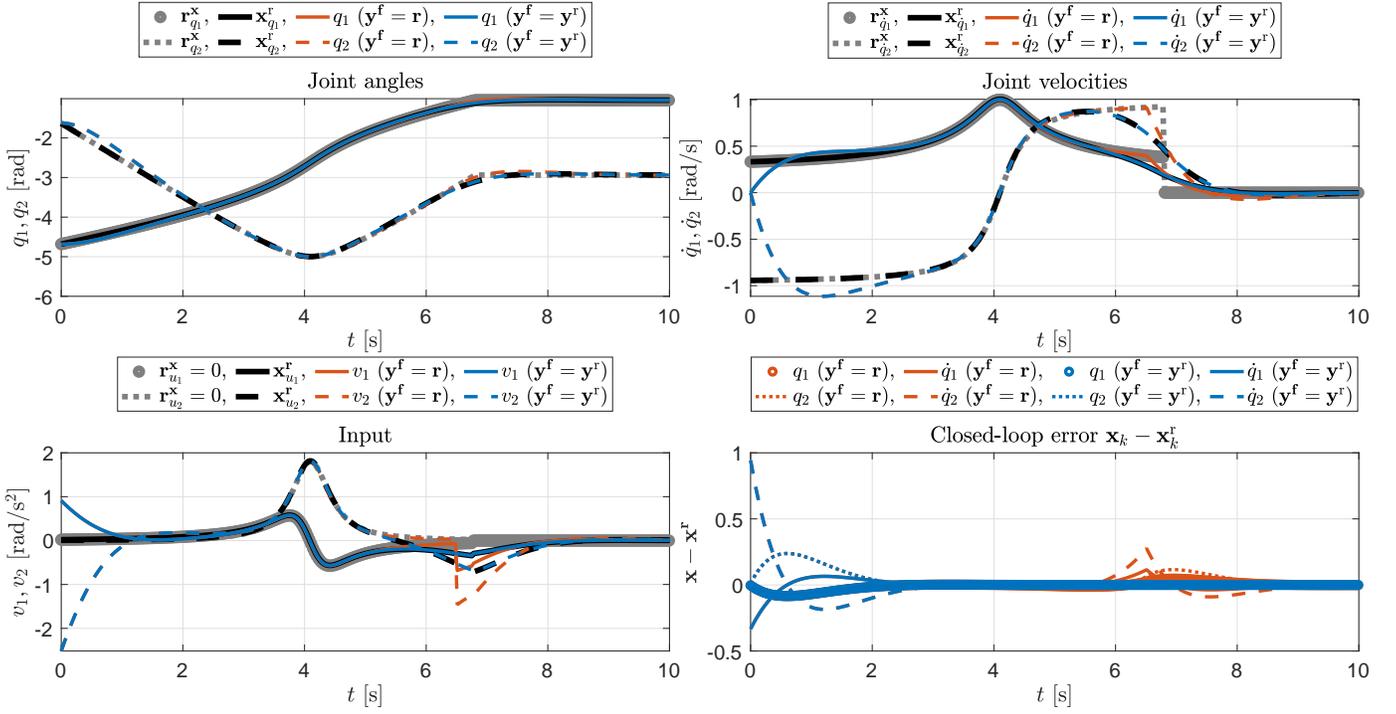}  
	\caption{Closed-loop simulation with initial condition $(x_1,x_2)=(-4.69,-1.62,0,0)$ and initial time $k=167$. The gray trajectories show the infeasible reference $\r=(\r^\x,\r^\u)$, while the black trajectories show the optimal reference $\y^\rr=(\x^\rr,\u^\rr)$ obtained from Problem~\eqref{eq:ocp}. The orange trajectories show the closed-loop behavior for the practical MPC Problem~\eqref{eq:nmpc}, while the blue trajectories show the closed-loop behavior for the \emph{ideal} MPC Problem~\eqref{eq:ideal_nmpc}.}
	\label{fig:mpatc_1_states}
\end{figure*}

	This theorem proves that one can use an infeasible reference, at the price of not converging exactly to the (unknown) optimal trajectory from OCP~\eqref{eq:ocp}, with an inaccuracy which depends on how inaccurate the terminal reference is. It is important to remark that, as proven in~\cite{Zanon2018a,Faulwasser2018}, since the MPC formulation has a turnpike, the effect of the terminal condition on the closed-loop trajectory is decreasing as the prediction horizon increases. 
	\begin{Remark}
		We note that similar results may be possible to prove for general nonlinear systems if there exists a storage function such that strict dissipativity holds for the rotated cost functions~\cite{muller2014necessity}. Future research will investigate ways to extend the results of Theorems~1 and~2 for general nonlinear systems.
	\end{Remark}

	\section{Simulations}\label{sec:simulations}
	In this section we implement the robotic example in~\cite{Faulwasser2009} to illustrate the results of Theorems~\ref{thm:as_stab_0} and \ref{thm:iss}. We will use the quadratic stage and terminal costs in \eqref{eq:stage_cost}-\eqref{eq:terminal_cost}, i.e.,
	\begin{gather*}
		q_\r(\xb,\ub,t_n) := \matr{c}{\xb-\rx_n\\\ub-\ru_n}^\top{}W\matr{c}{\xb-\rx_n\\\ub-\ru_n},\\
		p_\r(\xb,t_{n}) := (\xb-\rx_{n})^\top{}P(\xb-\rx_{n}).
	\end{gather*}

	We consider the system presented in~\cite{Faulwasser2009}, i.e., an actuated planar robot with two degrees of freedom with dynamics
	\begin{align}
		\matr{c}{\dot{x}_1\\\dot{x}_2} &= \matr{c}{ x_2\\B^{-1}(x_1)(u-C(x_1,x_2)x_2-g(x_1))},\label{eq:robot}
	\end{align}
	where $x_1=(q_1,q_2)$ are the joint angles, $x_2=(\dot{q}_1,\dot{q}_2)$ the joint velocities, and $B$, $C$, and $g$ are given by
	\begin{subequations}\label{eq:modelparams}
		\begin{align*}
			B(x_1) &:= \matr{cc}{200+50\cos(q_2) & 23.5+25\cos(q_2)\\ 
				23.5+25\cos(q_2) & 122.5},\\
			C(x_1,x_2) &:= 25\sin(q_2)\matr{cc}{\dot{q}_1 & \dot{q}_1+\dot{q}_2\\
				-\dot{q}_1	&  0}\\
			g(x_1) &:= \matr{c}{784.8\cos(q_1)+245.3\cos(q_1+q_2)\\
				245.3\cos(q_1+q_2)},
		\end{align*}
	\end{subequations}
	and with following constraints on the state and control 
	\begin{align}\label{eq:box_constr}
		\|x_2\|_\infty\leq{}3/2\pi, && \|u\|_\infty\leq{}4000.
	\end{align}
	By transforming the control input as
	$$u = C(x_1,x_2)x_2+g(x_1)+B(x_1)v,$$
	system~\eqref{eq:robot} can be rewritten into a linear system
	\begin{align}
		\matr{c}{\dot{x}_1\\\dot{x}_2} &= \matr{c}{ x_2\\v},\label{eq:robot_linear}
	\end{align}
	subject to the non-linear input constraint
	\begin{equation}
		\|C(x_1,x_2)x_2+g(x_1)+B(x_1)v\|_\infty\leq{}4000.
	\end{equation}

	Similar to~\cite{Faulwasser2009}, we use
	\begin{equation}\label{eq:path}
		p(\theta)=\left (\theta-\frac{\pi}{3},\,5\sin\left (0.6 \left (\theta-\frac{\pi}{3}\right )\right )\right ),
	\end{equation}
	with $\theta\in[-5.3,0]$ as the desired path to be tracked, and define the timing law, with $t_0=0\ \mathrm{s}$, to be given by
	\begin{align*}
		\theta(t_0) = -5.3,\, \dot{\theta}(t) = \frac{v_\mathrm{ref}(t) }{\left \| \nabla_\theta \rho(\theta(t))\right \|_2},\,  v_\mathrm{ref}(t) =\left \{ 
		\begin{array}{@{}ll@{}}
			1 & \hspace{-0.5em}\text{if } \theta<0\\
			0 & \hspace{-0.5em}\text{if }\theta\geq{}0
		\end{array}
		\right . .
	\end{align*}
	This predefined path evolution implies that the norm of the reference trajectory for the joint velocities will be  $1\ \mathrm{rad/s}$ for all $\theta<0$ and zero at the end of the path. Hence, we use the following reference trajectories
	\begin{align*}
		\r^\x(t) &= \matr{cc}{p(\theta(t))  &\frac{\partial{p}}{\partial\theta}\dot{\theta}(t)}^\top\hspace{-0.3em},\ 
		\r^\u(t) = \matr{c}{ \frac{\partial^2 p}{\partial\theta^2}\dot{\theta}^2+\frac{\partial p}{\partial \theta}\ddot{\theta}}^\top\hspace{-0.3em},
	\end{align*}
	which have a discontinuity at $\theta=0$.
	
	For the stage cost we use $W = \mathrm{blockdiag}(Q,R)$ with
	\begin{align*}
		Q=\mathrm{diag}(10,10,1,1),\
		R=\mathrm{diag}(1,1).
	\end{align*}
	The terminal cost matrix is computed using an LQR controller with the cost defined by $Q$ and $R$ and is given by
	$$ P = \matr{cc}{290.34\cdot{}\mathbf{1}^2 &105.42\cdot{}\mathbf{1}^2\\105.42\cdot{}\mathbf{1}^2&90.74\cdot{}\mathbf{1}^2}\in\mathbb{R}^4,$$
	where $\mathbf{1}^2\in\mathbb{R}^{2\times2}$ is an identity matrix. Consequently, the corresponding terminal set is then given by
	\begin{equation*}
		\mathcal{X}^\mathrm{f}_\r(t_n) =\{ \x\, |\, (\x-\r^\x_n)^\top P(\x-\r^\x_n) \leq{} 61.39\}.
	\end{equation*}
	For detailed derivations of the terminal cost and terminal set, we refer the reader to the Appendix in~\cite{Faulwasser2016,batkovic2020safe}.
	
	In order to obtain the feasible reference $\y^\rr=(\x^\rr,\u^\rr)$, we approximate the infinite horizon Problem~\eqref{eq:ocp} with a prediction horizon of $M=1200$ and sampling time $t_\mathrm{s}=0.03\ \mathrm{s}$. For the closed-loop simulations, we use the control input obtained from formulations~\eqref{eq:nmpc} and~\eqref{eq:ideal_nmpc} with horizon $N=10$ and sampling time $t_\mathrm{s}= 0.03\ \mathrm{s}$. Note that we used the linear system~\eqref{eq:robot_linear} with its corresponding state and input constraints for all problem formulations. Furthermore, all simulations ran on a laptop computer (i5 2GHz, 16GB RAM) and were implemented in Matlab using the CasADi~\cite{Andersson2019} software together with the IPOPT~\cite{wachter2006implementation} solver.
	
	Figure \ref{fig:mpatc_1_states} shows the closed-loop trajectories for the initial condition $(x_1,x_2)=(-4.69,-1.62,0,0)$ and initial time $k=167$. The gray lines denote the infeasible reference $\r=(\r^\x,\r^\u)$ for each state while the black lines denote the optimal reference $\y^\rr=(\x^\rr,\u^\rr)$ from~\eqref{eq:ocp}. The orange lines show the closed-loop evolution for the practical MPC Problem~\eqref{eq:nmpc}, i.e., when the terminal conditions are based on the infeasible reference $\y^\mathrm{f}=\r$. The blue lines instead show the closed-loop evolution for the \emph{ideal} MPC Problem~\eqref{eq:ideal_nmpc}, where the terminal conditions are based on the optimal reference from Problem~\eqref{eq:ocp}, i.e, $\y^\mathrm{f}=\y^\rr$. The bottom right plot of Figure~\ref{fig:mpatc_1_states} shows that the closed-loop error for both the practical MPC (orange lines) and \emph{ideal} MPC (blue lines) stabilize towards the reference $\r$ for times $t\leq{}5\mathrm{s}$. Between $5\ \mathrm{s}\leq{}t\leq{}9\ \mathrm{s}$, we can see that the discontinuity of the reference trajectory $\r$ affects how the two formulations behave. The \emph{ideal} formulation manages to track the optimal reference $\y^\rr$ (black trajectory), while the practical formulation instead tries to track the infeasible reference $\r$ and therefore deviates compared to the \emph{ideal} formulation. After the discontinuity, the rest of the reference trajectory is feasible and both formulations asymptotically stable.

	\section{Conclusions}\label{sec:conclusions}
	The use of infeasible references in MPC formulations is of great interest due to its convenience and simplicity. In this paper, we have discussed how such references affect the tracking performance for MPC formulations. We have proved that MPC formulations can yield asymptotic stability to an optimal trajectory when terminal conditions are suitably chosen. In addition, we also proved that the stability results can be extend for sub-optimal terminal conditions, in which case the controlled system is stabilized around a neighborhood of the optimal trajectory. Future research will investigate ways to extend the stability results to general nonlinear systems.

	\bibliographystyle{IEEEtran}
	\bibliography{references}
	
\end{document}